\documentclass{llncs}

\usepackage{paralist}
\usepackage{amsmath}
\usepackage{amsfonts}
\usepackage[sort]{cite}

\usepackage{listings}
\usepackage{algorithm,algorithmic}

\usepackage{tikz}
\usetikzlibrary{automata,shapes}
\tikzstyle{edge} = [draw,->]

\usepackage{hyperref}

\spnewtheorem*{example*}{Example}{\itshape}{\rmfamily}

\newcommand{\pponly}[1]{}
\newcommand{\rronly}[1]{#1}

\newcommand{\BB}{\mathbb{B}}

\newcommand{\QQ}{\mathbb{Q}}
\newcommand{\instance}[1]{\hat{#1}}
\newcommand{\eqvclass}[2]{\bar{#1}^{#2}}

\newcommand{\calM}{\mathcal{M}}
\newcommand{\calP}{\mathcal{P}}

\newcommand{\sleq}{\!\leq\!}

\newcommand{\defn}{\stackrel{\triangle}{=}}
\newcommand{\abstr}[1]{#1^\sharp}

\newcommand{\ie}{\textit{i.e.}}
\newcommand{\eg}{\textit{e.g.}}
\newcommand{\mat}[1]{{\boldsymbol #1}}
\renewcommand{\vec}[1]{{\boldsymbol #1}}

\newcommand{\verr}{\mathit{e}} 
\newcommand{\vheat}{\mathit{h}} 
\newcommand{\vfan}{\mathit{f}}

\author{David Monniaux\inst{1} \and Peter Schrammel\inst{2}}
\institute{CNRS / VERIMAG \and University of Oxford}
\title{Speeding Up \\ Logico-Numerical Strategy
  Iteration\rronly{\\(extended version)}%
\thanks{The research leading to these results has received funding
    from the European Research Council under the European Union’s
    Seventh Framework Programme (FP7/2007–2013) / ERC grant agreement
    306595 \href{http://stator.imag.fr/}{``STATOR''}
    and the ARTEMIS Joint Undertaking under grant
    agreement number 295311 \href{http://vetess.eu/}{``VeTeSS''}}}

\begin{document}
\maketitle

\begin{abstract}
We introduce an efficient combination of polyhedral analysis and predicate
partitioning.
Template polyhedral analysis abstracts numerical variables inside a
program by one polyhedron per control location, with \emph{a priori}
fixed directions for the faces. The strongest inductive invariant in
such an abstract domain may be computed by upward strategy
iteration. If the transition relation includes disjunctions and
existential quantifiers\rronly{ (a succinct representation for an exponential
set of paths)}, this invariant can be computed by a combination of
strategy iteration and satisfiability modulo theory (SMT) solving.
Unfortunately, the above approaches lead to unacceptable space and
time costs if applied to a program whose control states have been
partitioned according to predicates. We therefore propose a
modification of the strategy iteration algorithm where the strategies
are stored succinctly, and the linear programs to be solved at each
iteration step are simplified according to an equivalence
relation. 
We have implemented the technique in a prototype tool and we
demonstrate on a series of examples that the approach performs
significantly better than previous strategy iteration techniques.
\end{abstract}

\section{Introduction}
Program verification for unbounded execution times generally relies on
finding inductive loop (or procedure) invariants. In the
\emph{abstract interpretation} approach, loop invariants are
automatically searched within a class known as an \emph{abstract
  domain}. When dealing with numerical variables, it is common to
search for invariants shaped as products of intervals (constraints $l
\sleq x \sleq u$ with the program variable $x$ and bounds $l,u$),
convex polyhedra (constraint system $\mat{A}\vec{x} \sleq \vec{c}$
with the matrix $\mat{A}$, the vector of program variables $\vec{x}$,
and the vector of bounds $\vec{c}$), or restricted classes of convex
polyhedra such as \emph{octagons} (\rronly{constraints }$\pm x_i \pm x_j \sleq c$%
\rronly{ with program variables $x_i,x_j$ and a bound $c$}). Intervals and octagons
are instances of \emph{template polyhedra}: polyhedra where $\mat{A}$
is fixed \emph{a priori}, whereas in the general polyhedral approach,
$\mat{A}$ is discovered.  The restriction to fixed $\mat{A}$ reduces
the problem to finding suitable values for a fixed number of bounds
$\vec{c}$, and even, for certain classes of transitions, minimizing
these bounds using \emph{strategy iteration}
\cite{Gawlitza_Monniaux_LMCS12} (also known as \emph{policy
  iteration}) or other techniques, thereby producing the least (or
\emph{strongest}) inductive invariant in the abstract domain. In
contrast, for unknown $\mat{A}$ the number of constraints (rows in
the $\mat{A}$ matrix) may grow quickly and is most often limited by
\emph{widening} heuristics~\cite{DBLP:journals/fmsd/HalbwachsPR97}.

\rronly{
One common weakness of all these abstract domains is that they
represent only \emph{convex} numerical properties; it is for instance
impossible to represent $|x| \geq 1$ where $|x|$ denotes the absolute
value of $x$. An analyzer running on \\[0.5ex]
{\footnotesize
\textbf{while}(...) \{\\
\hspace*{1em}... \\
\hspace*{1em}\textbf{if} (abs(x) $>$= 10) \{ \textbf{assert}(x != 0); ... \} \\
\} \\[0.5ex]
}
\noindent will flag a warning on the assertion, because at this point
the non-convex property $x \leq -10 \lor x \geq 10$ has been
abstracted away.  In order to relieve this weakness, some recent
approaches
\cite{Henry_Monniaux_Moy_SAS2012,Gawlitza_Monniaux_LMCS12,Monniaux_Gonnord_SAS11}
advocate convex abstractions only at a \emph{cut-set} of all program
locations --- a subset such that removing all points in this subset
cuts all cycles in the control-flow graph, e.g. all loop heads within
a structured program. In between such distinguished locations, program
executions are exactly represented (or at least represented more
faithfully) as solutions to satisfiability modulo theory (SMT)
formulas. This is equivalent to replacing the original control flow
graph by a multigraph whose vertices are the distinguished nodes and
the edges are the simple paths between the distinguished nodes.

Variants on this basic idea include combinations with invariant
inference with widenings
\cite{Monniaux_Gonnord_SAS11,Henry_Monniaux_Moy_SAS2012} and with
strategy iteration \cite{Gawlitza_Monniaux_LMCS12}. With such
approaches, the test \lstinline|if (abs(x) >= 10)| is interpreted as
the disjunction $x \leq -10 \lor x \geq 10$ and the code is analyzed
in both contexts $x \leq -10$ and $x \geq 10$. If $n$ tests are used
in succession, the number of cases to analyze may be $2^n$, but such
methods eschew exhaustive enumeration for as-needed consideration of
paths inside the code through SMT-solving.

There still remains a difficulty: what if disjunctive invariants are
needed at the distinguished nodes?  
Assume for example, that predicate abstraction is used to 
handle programs that contain
constructs other than linear arithmetic, for example,
pointers, dynamic data stuctures, non-linear arithmetic, etc.
}
\pponly{
In order to handle programs that contain
constructs other than linear arithmetic, for example,
pointers, dynamic data stuctures, non-linear arithmetic, etc,
one can split control nodes according to $n$ predicates, as in
\emph{predicate abstraction}. The number of control nodes may thus
grow exponentially in~$n$.
}
A similar situation arises in
\emph{reactive programs} for control applications, where a main loop
updates global variables at each iteration, including Booleans\rronly{ (or,
more generally, variables belonging to a finite enumerated type)}
encoding ``modes'' of operation.
Such a system has a single distinguished control point (the head of
the main loop), yet, one wants to distinguish
invariants according to the mode of operation of the system.  Assuming
modes are defined by the values of the $n$ Boolean variables, this can
be achieved by splitting the loop head into $2^n$ distinct
control nodes and computing one invariant for each of them.  

Should we apply a max-strategy iteration modulo SMT algorithm
\cite{Gawlitza_Monniaux_LMCS12} to these $2^n$ control nodes, its
running time would be in the worst case proportional to $2^{d2^n}$
where $d$ is the number of disjuncts in the formula
defining the semantics of the program.  Worse, it would construct
linear programming problems with $2^n\ell$ unknowns, where $\ell$ is
the number of rows in the $\mat{A}$ matrix.
While high worst case complexity is not necessarily an objection (many
algorithms behave in practice better than their worst case),
constructing exponentially-\emph{sized} linear programs at every iteration of
the algorithm is certainly too costly. We thus previously left this partitioning
variant as an open problem~\cite[\S9]{Gawlitza_Monniaux_LMCS12}\cite[\S3.5]{DBLP:conf/vmcai/SchrammelS13}.

\paragraph{Contributions.}
The main contribution of this paper is an algorithm that computes the
\emph{same} result as these prohibitively expensive methods\rronly{ proposed in
\cite{GS07,Gawlitza_Monniaux_LMCS12}}, but limits the costs by
computing on-the-fly \rronly{a form of }equivalence between constraint bounds
(of which there are exponentially many) and constructing problems
whose size depends on the number of these equivalence classes. These
equivalence classes, in intuitive terms, distinguish Boolean variables
\emph{with respect to the abstraction chosen} (the $\mat{A}$
matrix). This is a novel aspect that distinguishes our approach from
quotienting techniques (\eg \cite{DBLP:journals/scp/BouajjaniFHR92}).
In contrast to \cite{DBLP:conf/vmcai/SchrammelS13} that uses
approximations to scale, we aim at computing the strongest invariant.
Finally, we show the results of an experimental evaluation conducted
with our prototype implementation that demonstrate the largely improved
performance in comparison to previous strategy iteration techniques.

\section{Strategy Iteration Basics}
Let us now recall the framework of strategy iteration over template
linear constraint domains \cite{GS07}, reformulating it to the setting of 
programs with linear arithmetic and Boolean
variables. As explained above, Boolean variables may be introduced by 
predicate abstraction or by the encoding of the control flow as in
reactive systems. Similar to \cite{Gawlitza_Monniaux_LMCS12}, 
this allows us to represent an exponential number of paths as a
single compact transition formula.
We then explain why previous algorithms
\cite{GS07,Gawlitza_Monniaux_LMCS12} \pponly{are unacceptably
  inefficient on exponentially many control nodes}\rronly{have unacceptable complexity when instantiated
  on our exponential number of control nodes}.

\paragraph*{Notation.} 
\rronly{We shall often talk both about formal variables appearing in logical
formulas and about individual values they may take, particularly those
obtained as satisfying instances of formulas; if needed, we shall }
\pponly{We }distinguish values by denoting them $\instance{x},\instance{y}\dots$
as opposed to variables $x,y,\dots$.
Variables $x_1,\dots,x_n$ are denoted collectively as a
vector~$\vec{x}$.  \rronly{When discussing satisfaction of logical formulas,
we shall note $(\vec{x},\vec{y}) \models F$ to mean explicitly that
$\vec{x},\vec{y}$ are free variables of $F$ that should satisfy~$F$.}
\pponly{$(\vec{x},\vec{y}) \models F$ means that
$\vec{x},\vec{y}$ are free variables of formula $F$ that should satisfy~$F$.}

\subsection{Program Model and Abstract Domain}
We model a program as a transition system with $m$ rational variables
$\vec{x}\rronly{ = (x_1,\dots,x_m)}\in \QQ^m$ (the \emph{numeric state}) and
$n$ Boolean variables $\vec{b}\rronly{ = (b_1,\dots,b_n) }\in \BB^n$ (the
\emph{Boolean state}), where $\BB = \{0,1\}$.  Let $I =
\rronly{(b_1^0,\dots,b_n^0, x_1^0,\dots,x_m^0) =} (\vec{b}^0, \vec{x}^0)$ be
the initial state.  The transition relation $\tau$ is of the form
$\exists y_1,\dots,y_E \in \QQ,~\exists p_1,\dots,p_d \in \BB.~T$
where $T$ is a quantifier-free formula in negation normal form, whose
atoms are either propositional ($b_i$, $\neg b_i$, $p_i$, $\neg p_i$),
linear (in)equalities ($\sum \alpha_i x_i + \sum \alpha'_i x'_i + \sum
\beta_i y_i \bowtie c$, where the $\alpha_i$, $\alpha_i'$, $\beta_i$
and $c$ are rational constants) with $\bowtie~\in\{\leq,<,=\}$, and
$y_i$ variables to encode nondeterminism\pponly{, e.g. reactive inputs
  or linearization of non-linear arithmetic};%
\rronly{\footnote{This limitation to linear (in)equalities may be
    lifted using, \eg, \emph{linearization}
    techniques~\cite{Mine_PhD04}.  For integer values, simple
    transformations should be performed, \eg, $x < y$ to $x \leq y-1$.
    Floating-point operations may also be relaxed to nondeterministic
    real operations~\cite[\S4.5]{Monniaux_LMCS10}.}}  the free
variables of $\tau$ are
\pponly{$\vec{x},\vec{x}',\vec{b}$}\rronly{$x_1,\dots,x_m,\allowbreak
  b_1,\dots,b_n, \allowbreak x'_1,\dots,x'_m,\allowbreak
  b'_1,\dots,b'_n$} where primed (respectively, unprimed) variables
denote the state after (respectively, before) the
transition. Furthermore, we add $p_i$ variables to each disjunction
with non-propositional literals, \ie, $x\leq 3 \vee x\geq 6$ becomes $(p_i
\land x\leq 3) \lor (\neg p_i \land x\geq 6)$. This encoding is necessary to
uniquely identify each disjunct by a Boolean proposition and extract
it from a SAT model.  The free variables of $T$ are thus grouped into
$\vec{b},\vec{b}',\vec{x},\vec{x}',\vec{p},\vec{y}$ where
$(\vec{b},\vec{x})$ and $(\vec{b}',\vec{x}')$ define respectively the
departure and arrival states and $\vec{p},\vec{y}$ stand for
intermediate values and choices.

\begin{example*}
We consider the following running example (a variant of the classical
thermostat model):

\lstset{numbers=left, numberstyle=\tiny, stepnumber=1, numbersep=5pt}
{\footnotesize
\begin{lstlisting}
bool error = 0, heat_on = 1;
bool fan_on = read_button();
real t = 16;
while(1) {
  real te = read_external_temp();
  assume(14<=te && te<=19);
  fan_on = read_button() ? !fan_on : fan_on;
  if(!error && (t<15 || t>30)) error = 1;
  else if(heat_on && t>22) heat_on = 0;
  else if(!error && heat_on && t<=22) t = (15*t + te)/16 + 1;
  else if(!error && !heat_on && t<18) heat_on = 1; 
  else if(!heat_on && t>=18) t = (15*t + te)/16;
}
\end{lstlisting}}
This program has the following transition relation $T$ with $n=3$
Boolean variables $\vec{b}=(\verr,\vheat,\vfan)$ (short for (error,heat\_on,fan\_on)),
$m=1$ numerical variables $\vec{x}=(t)$, $d=3$ path choice variables
$\vec{p}=(p_0,p_1,p_2)$), $\vec{y}=(te)$, and initial states $\neg \verr \wedge
\vheat \wedge t=16$:

\smallskip
{\small
$
\hspace*{-3ex}\begin{array}{c@{}c@{}c@{}c@{}c}
  \neg p_0 \wedge p_1 \wedge p_2 
  &\wedge&
  \neg \verr \wedge \verr' \wedge (\vheat=\vheat')
  &\wedge&
  t>30 \wedge t' = t~\vee\\

  \neg p_0 \wedge p_1 \wedge \neg p_2 
  &\wedge&
  \neg \verr \wedge \verr' \wedge (\vheat=\vheat')
  &\wedge&
  t<15 \wedge t' = t~\vee \\

  p_0 \wedge p_1 \wedge p_2
  &\wedge&
  \vheat \wedge \neg \vheat' \wedge (\verr=\verr')
  &\wedge&
  22<t\leq 30 \wedge t' = t~\vee \\

  p_0 \wedge p_1 \wedge \neg p_2
  &\wedge&
  \neg \verr \wedge \vheat \wedge \neg \verr' \wedge  \vheat' 
  &\wedge&
  t\leq 22 \wedge 14\leq te\leq 19 \wedge t' =
  \frac{15t+te}{16}+1~\vee \\
 
  p_0 \wedge \neg p_1 \wedge p_2
  &\wedge&
  \neg \verr \wedge \neg \vheat \wedge \neg \verr' \wedge \vheat'   
  &\wedge&
  15\leq t<18 \wedge t' = t~\vee \\

  p_0 \wedge \neg p_1 \wedge \neg p_2
  &\wedge&
  \neg \vheat \wedge \neg \vheat' \wedge  (\verr=\verr')
  &\wedge&
  t\geq 18 \wedge 14\leq te\leq 19 \wedge t' = \frac{15t+te}{16} 
\end{array}
$
}

\smallskip
\noindent The disjuncts stem from lines 9 --13; line 9 produces two
path choices. 
\end{example*}

\paragraph*{Abstract Domain.}
Let $\mat{A}$ be a $\ell \times m$ rational matrix, with rows
$\mat{A}_1,\dots,\mat{A}_\ell$.%
\rronly{\footnote{One can also make $\mat{A}$ depend on $b_1,\dots,b_n$ so as
  to apply a non-uniform abstraction, adding minor complication to
  algorithms and proofs. We chose to describe uniform abstraction for
  the sake of brevity and
  clarity.}} 
An element~$\rho$ of the abstract domain $\abstr{D}$ is a function
$\BB^n \rightarrow \overline{\QQ}^\ell$ with $\overline{\QQ} = \QQ \cup \{
-\infty, +\infty \}$. $\rho(\vec{b}) =\vec{c}$ means that 
at a Boolean state $\vec{b}$ the vector of numerical variables $\vec{x}$ is such
that $\mat{A}\vec{x} \sleq \vec{c}$ coordinate-wise.
Moreover, we write $\rho(\vec{b}){=}-\vec{\infty}$ if 
any coordinate $c_i{=}-\infty$, meaning that the Boolean state $\vec{b}$ is
unreachable (because $\mat{A}\vec{x} \sleq \vec{c}$ is $\mathit{false}$). We note
$\gamma(\rho)$ the set of states $(\vec{b},\vec{x})$ verifying these
conditions.  $\QQ^\ell$ is ordered by
coordinate-wise $\leq$, inducing a point-wise ordering $\sqsubseteq$
on $\abstr{D}$. $\gamma$ is thus monotone w.r.t.
$\sqsubseteq$ and the inclusion ordering on sets of states\rronly{; note that
it is not injective in general}.  We denote by $\rho(i,\vec{b})$ the
$i$-th coordinate of $\rho(\vec{b})$.
$\rho$ is said to be an \emph{inductive invariant} if it contains the initial state ($\mat{A}\vec{x}^0 \leq \rho(\vec{b}^0)$) and it is stable by transitions:\\
\centerline{$\forall \vec{b},\vec{x},\vec{b}',\vec{x}'.~ (\vec{b},\vec{x}) \in \gamma(\rho) \land (\vec{b},\vec{x},\vec{b}',\vec{x}') \models \tau \implies (\vec{b}',\vec{x}') \in \gamma(\rho')$,}\rronly{\\
otherwise said
$\forall \vec{b},\vec{x},\vec{b}',\vec{x}'.~ \rho(\vec{b}) \neq \bot \land \mat{A}\vec{x} \leq \rho(\vec{b})\land (\vec{b},\vec{x},\vec{b}',\vec{x}') \models\tau$\\
\centerline{$\implies \rho(\vec{b}') \neq \bot \land \mat{A}\vec{x}' \leq \rho(\vec{b}')$.}}

The main contribution of this paper is an effective way to compute the
least inductive invariant $\rho$ in this abstract
domain\pponly{ w.r.t. $\sqsubseteq$}\rronly{ with
  respect to inclusion ordering}.

\subsection{Strategy Iteration}
\label{sec:original_strategy_iteration}
Recall that the original strategy iteration algorithm \cite{GS07}
applies to disjunctive systems of linear inequalities (of exponential
size in $d$) induced by the collecting semantics of the program over
template polyhedra (see Equ.~(\ref{equ:transformer}) below).
Previous work \cite{Gawlitza_Monniaux_LMCS12} improves the algorithm by keeping
 the system implicit, only extracting a linear size system at any
 given time using SMT solving.
Note that $\tau$, after replacing each free Boolean variable by a
Boolean constant, is equivalent to a disjunction of (at most) $2^d$
formulas of the form $\exists \vec{y}.C$ where $C$ is a
conjunction of non-strict linear inequalities, and $d$ is the number
of Boolean existential quantifiers in~$\tau$.
In both algorithms, a \emph{strategy}%
\rronly{\footnote{The word ``strategy'' (or ``policy'') arises from an
    analogy between the system of min-max + monotone affine linear
    equalities whose least solution yields the least invariant in the
    domain \cite{Gawlitza_Monniaux_LMCS12,GS07} and the system of
    min-max + barycentric inequalities whose least solution is the
    value of a two-player Markov game.}}  
selects one disjunct in $C$ for each template row $i'$.  
Hence, we can use these algorithms in our setting by selecting a
disjunct for each template row \emph{and} each Boolean valuation for
$(\vec{b'},\vec{p})$.  This motivates the following definition:

A \emph{strategy} associates with each Boolean state $\vec{b}'$
and each constraint index $1 \sleq i' \sleq \ell$ either the special
value $\bot$, meaning that $\vec{b}'$ is unreachable and denoted by
$\sigma(i',\vec{b}')=\pi(i',\vec{b}')=\bot$, or a pair
$\big(\sigma(i',\vec{b}'),\pi(i',\vec{b}') \big)$ where
$\sigma(i',\vec{b}') \in \BB^n$ is a Boolean state and
$\pi(i',\vec{b}') \in \BB^d$ gives ``path choices''.  Once $\pi \in
\BB^d$ is chosen, the result of substituting $T[\pi/\vec{p}]$ is a
conjunction of linear inequalities and a propositional formula (in
the variables $\vec{b},\vec{b}'$); let
$T_\pi$ denote the conjunction of these linear inequalities.\pponly{\\[-3ex]}%
\rronly{\footnote{This conjunction corresponds to a \emph{merge-simple
    statement} of \cite{Gawlitza_Monniaux_LMCS12} or to a \emph{path}
  of \cite{Henry_Monniaux_Moy_SAS2012,Monniaux_Gonnord_SAS11}.}}

\paragraph*{Algorithm.}
Let us now see how the algorithm iterates until the least inductive
invariant is reached.
The algorithm maintains, at iteration number $k$, a current strategy
$(\sigma_k,\pi_k)$.  Initially, the abstract value $\rho_0$ is $\bot$
everywhere save at the initial Boolean state $\vec{b}^0$, where
$\rho(\vec{b}^0)=\mat{A}\vec{x}^0$; $\sigma_0$ and $\pi_0$ are $\bot$
everywhere. For $k \geq 0$, the strategy yields $\rho_{k+1}$ as the
least fixed point $\mu_{\geq \rho_k} \Psi_{\pi}$ greater than
$\rho_k$, $\Psi_{\pi}$ being an order-continuous operator on the
lattice $(\BB^n \times \{ 1, \dots, \ell\}) \rightarrow
\overline{\QQ}$ defined as: \\[-2ex]
\begin{equation}\label{equ:transformer}
\Psi_{\pi}(\rho) \defn (i',\vec{b}') \mapsto
  \sup \left\{ \mat{A}_i \vec{x}' \mid \exists \vec{x},\vec{y}.~T_{\pi(i',\vec{b}')}
          \land (\mat{A}\vec{x} \leq \rho(\sigma(i',\vec{b}'))) \right\}
\end{equation}
We explain in \S\ref{sec:original_value} how to compute this fixed
point; let us now see how $\sigma_{k+1}$ and $\pi_{k+1}$ are obtained
from $\sigma_k$ and $\pi_k$, and how termination is
decided~\cite[\S6.5]{Gawlitza_Monniaux_LMCS12}.  Each iteration goes
as follows: for all Boolean states $\instance{\vec{b}'} \in\BB^n$ and
all $\instance{\vec{b}}$ with $\rho_k(\instance{\vec{b}})\neq-\vec{\infty}$:\pponly{\\[-3ex]}
\begin{enumerate}
\item construct formula
  $T[\instance{\vec{b}}/\vec{b},\instance{\vec{b}'}/\vec{b}']$, that
  is, $T$ where variables $\vec{b}$ and $\vec{b}'$ have been replaced
  by Boolean values $\instance{\vec{b}}$ and $\instance{\vec{b}}'$;
\item conjoin it with constraints $\mat{A}\vec{x} \leq
  \rho_k(\vec{b})$ and $\mat{A}_i\vec{x} > \rho_k(i,\vec{b}')$,
thus obtaining\pponly{\\[-4ex]}
\begin{equation}\label{for:invariant_violation}
  T[\instance{\vec{b}}/\vec{b},\instance{\vec{b}'}/\vec{b}'] \land \mat{A}\vec{x} \leq \rho_k(\instance{\vec{b}}) \land \mat{A}_i\vec{x}' > \rho_k(i,\instance{\vec{b}'})
\end{equation}
\item check whether this formula (in free variables
  $x_1,\dots,x_m,\allowbreak x'_1,\dots,x'_m,\allowbreak
  y_1,\dots,y_E, \allowbreak p_1,\dots,p_d$) is satisfiable;
\item if this formula is satisfiable, $\rho_k$ does not describe an
  inductive invariant: the satisfying instance describes a transition
  from a state $(\vec{b},\vec{x})$ to a state $(\vec{b}', \vec{x}')$
  such that $(\vec{b},\vec{x})$ lies within the invariant but
  $(\vec{b}',\vec{x}')$ does not; this solution yields a new strategy
  $\pi_{k+1}(i,\instance{\vec{b}'})=\instance{\vec{p}}$ and
  $\sigma_{k+1}(i,\instance{\vec{b}'})=\instance{\vec{b}}$, which
  \emph{improves} on the preceding
  one~\cite[\S6.3]{Gawlitza_Monniaux_LMCS12};
\item if $\pi_{k+1} (i,\instance{\vec{b}'})$ and
      $\sigma_{k+1} (i,\instance{\vec{b}'})$ are not set by the preceding
      step, leave them to the their previous values $\pi_k(i,\instance{\vec{b}'})$ and $\sigma_k(i,\instance{\vec{b}'})$; if none have been
      updated, this means $\rho_k$ is the least inductive invariant, thus
      \textbf{exit}; 
\item otherwise, compute $\rho_{k+1} = \mu_{\geq \rho_k} \Psi_{\pi_{k+1}}$ and continue iterating.
\end{enumerate}
The main loop of this algorithm enumerates each of the $2^{(n+d)\ell 2^n}$
strategies at most once.  Remark the important improvement condition:
at every iteration but the last, $\Psi_\pi(\rho) > \rho$. Since there
is a finite number of strategies that may deem $\rho$
non-inductive and each of them is chosen at most once, we are
guaranteed to terminate with the least fixed point (without using any
widening) within a finite number of steps.

\begin{example*}
Let us analyze our running example using the box template $(t,-t)^T$.
Assume the current abstract value\footnote{For better readability, we write,
for example, $\bar{\verr}\vheat\vfan$ for
the value $(0,1,1)$ of $(\verr,\vheat,\vfan)$.} 
$\rho_k(i,\bar{\verr}\vheat\vfan)=16$ for $i\in\{1,2\}$.
To compute an improved strategy $\sigma_{k+1},\pi_{k+1}$, we have to check all values of
$(\instance{\vec{b}},\instance{\vec{b}'})$, \eg, 
$(\bar{\verr}\vheat\vfan,\bar{\verr}'\vheat'\vfan')$:
 instantiating Equ.~\ref{for:invariant_violation} with these values
 (with $T$ from our running example and $i=1$) gives\\[0.6ex]
\centerline{\small
$\begin{array}{c}
  \big(p_0 \wedge p_1 \wedge \neg p_2
\wedge
  t\leq 22 \wedge 14\leq te\leq 19 \wedge (t' = \frac{15t+te}{16}+1)\big)
\wedge 
\big(t=16\big) 
\wedge
\big(t'>16\big)
\end{array}
$}\\[0.6ex]
which is satisfied, for instance, by the model
$(\instance{\vec{x}},\instance{\vec{x}}',\instance{\vec{p}})=(16,17,(1,1,0))$.  Hence, we
update the strategy by setting $\sigma_{k+1}(1,\bar{\verr}\vheat\vfan)=\bar{\verr}\vheat\vfan$
and \\ $\pi_{k+1}(1,\bar{\verr}\vheat\vfan)=(1,1,0)$,
which induces\\[0.6ex]
\centerline{\small $
\begin{array}{c}
T_{\pi_{k+1}(1,\bar{\verr}\vheat\vfan)} = (t\leq 22 \wedge 14\leq
te\leq 19 \wedge (t' = \frac{15t+te}{16}+1))
\end{array}.
$ }

After having checked all $(\instance{\vec{b}},\instance{\vec{b}'})$,
we can compute the strategy value, \ie, 
the fixed point of $\Psi_{\pi_{k+1}}$, which updates
$\rho_{k+1}(1,\bar{\verr}\vheat\vfan)$ to
$\frac{365}{16}$ in this case.\footnote{By maximizing $t$ for any $te$, we get
  $\frac{15\cdot 22+19}{16}+1=\frac{365}{16}$.}  The way this is
computed is explained in the next section.
\end{example*}

\subsection{Computing the Strategy Value}
\label{sec:original_value}
We recall now how to compute the strategy fixed point $\mu_{\geq \rho}
\Psi_{\pi}$~\cite[\S6.4]{Gawlitza_Monniaux_LMCS12}, under the
condition that $\Psi_\pi(\rho) \geq \rho$ (which is always the case,
because of the way $\pi$ is chosen).

The first step is to identify the Boolean states $\vec{b}$
``abstractly unreachable'': such $\vec{b}$ form the least set $Z$
containing all $\vec{b} \neq \vec{b}^0$ such that
$\pi(i,\vec{b})=\bot$ and stable by: if $\vec{b}' \neq \vec{b}^0$ is
such that $\sigma(i,\vec{b}') \in Z$ then $\vec{b}' \in Z$; for all
$\vec{b} \in Z$, set $\rho(\vec{b}):=-\vec{\infty}$.

Construct a system of linear inequalities in the unknowns
$v_{i,\vec{b}}$ for $\vec{b} \in \BB^n$ and $1 \sleq i \sleq \ell$, plus
fresh variables: for all $\vec{b}' \notin Z$, for all $1 \leq i' \leq
\ell$ such that $\rho(i',\vec{b}') < +\infty$, add the inequalities\\
\hspace*{1em}$
\begin{array}{ll@{\hspace{1.5em}}l}
\bullet & \mat{A}_j\vec{x} \leq v_{j,\sigma(i',\vec{b}')}\text{ for all }1 \leq
j \leq \ell & \text{\small(``in departure state invariant'')} \\
\bullet & \mat{A}_{i'}\vec{x}' \geq v_{i',\vec{b}'} & 
\text{\small(``bounded by arrival state'')}\\
\bullet & \text{those from the conjunction }T_{\pi(i',\vec{b}')} &
\text{\small(``follows the transition relation'')}
\end{array}$\\
where variables $\vec{x}$ and $\vec{y}$ have been replaced
by fresh variables (each different $i,\vec{b}'$ has its own set of
fresh replacements).  $\rho(i,\vec{b}')$ is obtained by linear
programming as the maximum of $v_{i',\vec{b}'}$ satisfying this
system.  This linear program has solutions, otherwise the strategy
$\sigma,\pi$ would not have been chosen; if it has no \emph{optimal} solution
it means that $\rho(i',\vec{b}')=+\infty$.

Note that these $O(2^n \ell)$ linear programs have \pponly{$O(2^n \ell)$}%
\rronly{$O((2\ell+E)2^n)$} variables and a system of inequalities of size
$O(2^n |T|)$ where $|T|$ is the size of formula $T$. It is in fact
possible to replace these $O(2^n \ell)$ linear programs by two linear
programs of size \pponly{$O(2^n \ell)$}\rronly{$O((2\ell+E)2^n)$}: first,
one using the $\infty$-abstraction (see
\cite[\S8,9]{Gawlitza:2011:SSR:1961204.1961207}) to obtain which of
the $v_{i',\vec{b}'}$ go to $+\infty$, then another for maximizing
$\sum v_{i',\vec{b}'}$ restricted to the $v_{i',\vec{b}'}$ found not
to be $+\infty$ by the $\infty$-abstraction.

\section{Our Algorithm}
\label{sec:our_algorithm}
Notice three difficulties in the preceding algorithm: there are,
\emph{a priori}, $2^{2n}$ SMT-solving tests to be performed at each
iteration; the linear programs have exponential size; and there are at
most $2^{(n+d)\ell 2^n}$ strategies, thus a doubly exponential bound on the
number of iterations.
In intuitive terms, the first two difficulties stem from the explicit
expansion of the exponential set of Boolean states, despite the
implicit representation of the exponential set of execution paths
between any two control (Boolean) states $\vec{b}$ and $\vec{b}'$, a
weakness that we shall now remedy.

\subsection{Strategy Improvement Step}
\label{sec:new_strategy_improvement}

The first difficulty is the easiest to solve: the $2^{2n}$ SMT-tests,
one for each pair $(\vec{b},\vec{b} ')$ of control states, can be
folded into one single test where the $\vec{b}$ and $\vec{b}'$ also
are unknowns to be solved for.

Note that the structure of 
$\rho$, $\BB^n \rightarrow \overline{\QQ}^\ell$,
can be viewed as 
$\{1,\ldots,\ell\}\rightarrow (\BB^n \rightarrow \overline{\QQ})$.
Hence, we need not store a $2^n \times \ell$ array of rationals (or
infinities), but we can implement it efficiently as an array (of size
$\ell$) of \textsc{Mtbdd}s \cite{DBLP:journals/tc/Bryant86} with the bounds $c_{i,j}$ in the leaves.
Assume for a given template row $i$, we have $s_i$ different bounds
$c_{i,j}$, and denote $\phi_{i,j}$ the propositional formula
describing the set of Boolean states that map to bound $c_{i,j}$.
Then, observe that $\phi_{i,j}$ for $1\sleq j\sleq s_i$ form a
partition of $\BB^n$\rronly{ (that is, $\bigvee_{j=1}^{s_i} \phi_{i,j}$ is a tautology and each $\phi_{i,j} \land \phi_{i,k}$, $j \neq k$, is unsatisfiable)}.
We use the notation $\rho(i) = \{\phi_{i,1} \rightarrow
c_{i,1},\ldots,\phi_{i,s_i} \rightarrow c_{i,s_i}\}$ to represent an
\textsc{Mtbdd}, and $\rho(i,\vec{b}) = c_{i,j}$ to obtain the bound
$c_{i,j}$ for state $\vec{b}$ for template row $i$.\pponly{\\[-4ex]}

\paragraph{Strategy improvement condition.}
In Equ.~\ref{for:invariant_violation},
one may replace $\mat{A}\vec{x} \leq \rho(\vec{b})$ and $\mat{A}_i
\vec{x} > \rho(i',\vec{b} ')$ respectively by $\psi_1$ and $\psi_2$:
\pponly{\vspace*{-1ex}}
\begin{eqnarray}
\psi_1 \defn \bigwedge_i \bigvee_{j=1}^{s_i} \phi_{i,j}(\vec{b})
 \land \mat{A}_i\vec{x} \leq c_{i,j} \label{for:invariant_violation2a}\\[-1ex]
\psi_2 \defn \bigvee_{j=1}^{s_i} \phi_{i',j}(\vec{b}') \land \mat{A}_i\vec{x}' = c_{i',j} + \Delta \land \Delta > 0\label{for:invariant_violation2b}
\end{eqnarray}
\pponly{\vspace*{-3ex}}

Remark that $\psi =_{\mathit{def}} \psi_1 \land T \land \psi_2$ is satisfiable iff
there is a transition from $(\vec{b},\vec{x})$ inside the
invariant defined by $\rho$ to $(\vec{b}',\vec{x}')$ outside of it.
\rronly{
The same applies if we replace $T$ in $\psi$ by a \emph{slicing} or
\emph{cone of influence} $T_i$ of $T$ with respect to the value of
$\mat{A}_i\vec{x}$, that is, a formula $T_i$ such that $(\exists p_1,
\dots, p_d \in \BB, ~ \exists y_1, \dots, y_E \in \QQ.~T) \land
\mat{A}'_i \vec{x}' \geq v$ and $(\exists p_1, \dots, p_d \in \BB, ~
\exists y_1, \dots, y_E \in \QQ.~T_i) \land \mat{A}'_i \vec{x}' \geq
v$ are equivalent (w.r.t $\vec{b},\vec{x},\vec{b}',v$).  When $T$ is
compiled from a program, such a $T_i$ may be obtained using program
slicing.}
The strategy iteration algorithm progresses regardless of $\Delta$ as
long as $\Delta>0$. 
\rronly{It is however likely that maximizing $\Delta$
leads to faster convergence \cite[p.~26]{Gawlitza_Monniaux_LMCS12},
because there is no backtracking in max-strategy iteration and hence a
\emph{locally optimal}, \ie, greedy, strategy cannot be
disadvantageous.  Maximizing $\Delta$ may be performed by
\emph{optimization modulo theory}
techniques~\cite{DBLP:conf/sat/NieuwenhuisO06,DBLP:conf/cade/SebastianiT12,LAK+14}.}

Obtaining a solution
$\vec{b},\vec{x},\vec{b}',\vec{x}',\vec{y},\vec{p} \models \psi$
enables us to improve the strategy by setting
$\sigma(i,\vec{b}'):=\vec{b}$ and $\pi(i,\vec{b}'):=\vec{p}$, as in
\S\ref{sec:original_strategy_iteration}.

\begin{example*}
Let us assume the following current abstract value in the analysis of
our running example:\\[0.6ex]
\centerline{\small $
\begin{array}{ll}
\rho(1)=\{\neg\verr\wedge\vheat \rightarrow 16, &\verr \vee \neg\vheat\rightarrow -\infty\}\\
\rho(2)=\{\neg\verr\wedge\vheat \rightarrow -16, &\verr \vee
\neg\vheat\rightarrow -\infty\}
\end{array}
$}

\noindent We build Equ.~\ref{for:invariant_violation} using 
Equs.~\ref{for:invariant_violation2a} and
\ref{for:invariant_violation2b}:\\[0.6ex]
\centerline{\small $
\psi = T\wedge \left(\begin{array}{lcl}
\big(
\neg\verr\wedge\vheat \wedge t\leq 16
&\vee& (\verr \vee \neg\vheat)\wedge t\leq-\infty
\big) \wedge\\
\big(
\neg\verr\wedge\vheat \wedge -t\leq -16 &\vee&
(\verr \vee \neg\vheat)\wedge -t\leq-\infty
\big) \wedge\\
\big(
\neg\verr'\wedge\vheat' \wedge (t'=16+\Delta) &\vee& 
(\verr' \vee \neg\vheat')\wedge (-t'=-\infty+\Delta)\big) 
\end{array}\right) \wedge
\Delta>0
$}

\noindent This formula is satisfied, \eg, by the model
$(\instance{\vec{b}},\instance{\vec{b}'},\instance{\vec{x}},\instance{\vec{x}}',\instance{\vec{p}})=
(\bar{\verr}\vheat\vfan,\bar{\verr}\vheat\vfan,$ $16,17,(1,1,0))$.
Hence, we update
$\sigma(1,\bar{\verr}\vheat\vfan):=\bar{\verr}\vheat\vfan$ and
$\pi(1,\bar{\verr}\vheat\vfan):=(1,1,0)$.  We have to repeat this
check excluding the above solution to find other models, \eg,
$(\bar{\verr}\vheat\bar{\vfan},\bar{\verr}\vheat\vfan,16,17,(1,1,0))$.
\end{example*}

Improving the strategy this way would however be costly, since we
would have to do it one $\instance{\vec{b}}'$ at a time (by naive model
enumeration). 

\paragraph{Model generalization.}
There is however a better way by \emph{generalizing} from an obtained
model to a set of $\instance{\vec{b}}'$ that can be updated at once:
Notice now that, fixing $\instance{\vec{x}}$ and $\instance{\vec{y}}$
arising from a solution,
$\psi[\instance{\vec{x}}/\vec{x},\instance{\vec{y}}/\vec{y}]$ becomes
a purely propositional formula, whose models also yield suitable
solutions for $\vec{b},\vec{b}',\vec{p}$.
Fix $\instance{\vec{b}}$
and $\instance{\vec{p}}$ from a solution, then the free variables are
now only the $\vec{b}'$; then for any solution $\instance{\vec{b}'}$
of
$\psi[\instance{\vec{x}}/\vec{x},\instance{\vec{y}}/\vec{y},\instance{\vec{b}}/\vec{b},\instance{\vec{p}}/\vec{p}]$,
we can set $\pi(\instance{\vec{b}'},i):=\instance{\vec{p}}$ and
$\sigma(\instance{\vec{b}'},i):=\instance{\vec{b}}$.
We can thus improve strategies for whole sets of $\instance{\vec{b}}'$ at once in nondeterministic systems.

\begin{algorithm}[tb]
\caption{\textsc{Improve}: Selecting the strategy improvement
\label{algo:improve}}
\begin{algorithmic}[1]
\STATE{$\mathit{stable} := \mathit{true}$}
\FOR{$i' \in \{ 1, \dots, \ell\}$}
  \STATE{$U := \mathit{false}$  \hspace{2em} // $U$ defines the set of $\vec{b}'$ such that $\pi(i',\vec{b}')$ has
been updated.}
  \WHILE{\hfill$\neg U \land \left(\bigwedge_i \bigvee_{j=1}^{s_i} \phi_{i,j}(\vec{b}) \land \mat{A}_i\vec{x} \leq c_{i,j}\right) 
    \land T_{i'}\land
    \left(\bigvee_{j=1}^{s_i} \phi_{i,j}(\vec{b}') \land \mat{A}_i\vec{x}' = c_{i,j} + \Delta\right) \land \Delta > 0$ is satisfiable}
    \STATE{$\langle \instance{\vec{b}},\instance{\vec{x}},\instance{\vec{b}'},\instance{\vec{x}'},\instance{\vec{p}},\instance{\vec{y}} \rangle :=$ a model of the above formula \rronly{(optionally of max. $\Delta$)}}
    \STATE{$F := T_i[\instance{\vec{x}}/\vec{x},\instance{\vec{y}}/\vec{y}] \land \neg U$}
    \STATE{$\mathit{stable} := \mathit{false}$}
    \WHILE{$F$ is satisfiable}
      \STATE{$\langle \instance{\vec{b}_1}, \instance{\vec{b}'_1}, \instance{\vec{p}_1} \rangle :=$ a model of $F$}
      \STATE{$G := F[\instance{\vec{b}_1}/\vec{b},\instance{\vec{p}_1}/\vec{p}]$}
      \STATE{$F := F \land \neg G$}
      \STATE{$\pi[i',G] := \instance{\vec{p}}$ \hspace{2em} // $\pi[i',G] := \instance{\vec{p}}$ means ``in the mapping $\vec{b}' \mapsto \pi(i',\vec{b}')$,}
      \STATE{$\sigma[i' ,G] := \instance{\vec{b}}$ \hspace{2.04em} // \; replace all images of $\vec{b}'$ satisfying
formula $G$ }
      \STATE{$U := U \lor G$ \hspace{2.08em} // \; by $\instance{\vec{p}}$'' (respectively for~$\sigma$).}
    \ENDWHILE
  \ENDWHILE
\ENDFOR
\end{algorithmic}
\end{algorithm}
\begin{algorithm}[tb]
\caption{\textsc{Iterate}: Main strategy iteration algorithm
\label{algo:iterate}}
\begin{algorithmic}
\FOR{$i \in \{ 1, \dots, \ell\}$}
  \STATE{$\phi_{1,i} := (\vec{b}=\vec{b}^0)$;\quad
         $c_{1,i} := \mat{A}_i \vec{x}^0$;\quad 
         $\phi_{2,i} := (\vec{b} \neq \vec{b}^0)$;\quad
         $c_{2,i} := -\infty$}
\ENDFOR
\STATE{$\mathit{stable} := \mathit{false}$}
\WHILE{$\neg \mathit{stable}$}
  \STATE{\textsc{Improve}}
  \IF{$\neg \mathit{stable}$}
     \STATE{\textsc{Compute-Strategy-Value} $\quad$ (see \S\ref{sec:new_strategy_value})}
  \ENDIF
\ENDWHILE
\end{algorithmic}
\end{algorithm}

Our strategy improvement algorithm (procedure \textsc{Improve},
Alg.~\ref{algo:improve}) thus proceeds as follows: it maintains a set
$U$ of ``already improved'' values of $\vec{b}'$, and requests
$(\vec{b},\vec{b}',\vec{p})$ by SMT-solving as described above, with
the additional constraint that $\vec{b}' \notin U$; if no such
solution is found, it terminates, having done all improvements,
otherwise it generalizes $\vec{b}'$ to a whole set of solutions as
described above, and improves the strategy for all
these~$\vec{b}'$. The strategy $\pi,\sigma$ and the set $U$ are stored
in \textsc{Bdd}s.

\begin{example*}
Let us assume we have the current abstract value\\
\centerline{$
\begin{array}{ll}
\rho(1)=\{\neg\verr\wedge\vheat \rightarrow \frac{365}{16}, &\verr\vee \neg\vheat\rightarrow -\infty\}\\
\rho(2)=\{\neg\verr\wedge\vheat \rightarrow -16, &\verr\vee \neg\vheat\rightarrow -\infty\}.
\end{array}
$}
Moreover, assume that we have obtained the model of
$\psi$: 
$(\instance{\vec{b}},\instance{\vec{b}'},\instance{\vec{x}},\instance{\vec{x}}',\instance{\vec{p}})=
(\bar{\verr}\vheat\vfan,\bar{\verr}\bar{\vheat}\vfan,$ $\frac{365}{16},\frac{365}{16},(1,1,1))$.
Substituting the values of this solution for $\vec{x}$ and $\vec{x}'$ in formula
$\psi$, we get 
$F=p_0\wedge p_1\wedge p_2\wedge \neg\verr \wedge \vheat\wedge
\neg\verr' \wedge \neg\vheat'$.  
Now, we substitute the above values for $\vec{b}$ and $\vec{p}$ in $F$,
which gives us $G=\neg\verr' \wedge \neg\vheat'$. We update the strategy
$\sigma,\pi$ for
the whole set of states satisfying $G$, \ie,
$\{\bar{\verr}\bar{\vheat}\vfan,\bar{\verr}\bar{\vheat}\bar{\vfan}\}$
at once, and we add $G$ to $U$.
 Then we ask the SAT solver again for a model of the formula
$F\wedge\neg G$, which is unsatisfiable in this example.
We continue enumerating the solutions of $\psi$, but this time 
excluding $U$, \ie, we call the SMT solver with $\psi\wedge\neg U$, 
which is unsatisfiable in our example. Hence, we have completed 
strategy improvement for the first template row.
For row 2, we proceed similarly and obtain the same strategy update.
The associated strategy value computation yields the abstract value
$\rho$:\\[0.6ex]
\centerline{$
\begin{array}{lll}
\rho(1)=\{\neg\verr \rightarrow \frac{365}{16}, &&\verr\rightarrow -\infty\}\\
\rho(2)=\{\neg\verr\wedge\vheat \rightarrow -16, &\neg\verr\wedge \neg\vheat \rightarrow -22, &\verr\rightarrow -\infty\}.
\end{array}
$}
\end{example*}
\pponly{\vspace*{-4ex}}
\rronly{
\begin{lemma}\label{lemma:improve}
\textsc{Improve} terminates in at most exponential time. At the end,
$\mathit{stable}$ is false if and only if the strategy needed updating
(otherwise said, $\gamma(\rho_k)$ was not an inductive invariant), in
which case $\sigma,\pi$ contain the next strategy
$\sigma_{k+1},\pi_{k+1}$.
\end{lemma}
\begin{proof}
Each iteration of the outer (resp. inner) loop removes at least one
solution for $\vec{b}'$ from $\neg U$ (resp.~$F$), and there are $2^n$
of them. The updates to $\pi$ and $\sigma$ have been explained in the
preceding paragraphs.
\end{proof}}
\begin{theorem}
\textsc{Iterate} (Alg.~\ref{algo:iterate}) terminates in at most
$2^{(n+d)m2^n}$ iterations, with the final $\rho$ being equal to that
computed by the algorithm of \S\ref{sec:original_strategy_iteration},
yielding the least inductive invariant in the domain.%
\pponly{\footnote{
Due to space limitations, we refer to the extended version
\cite{MS14a} for the proofs.}}
\end{theorem}
\rronly{
\begin{proof}
The correctness of the result ensues from the correctness of the
path-focused strategy iteration approach
\cite{Gawlitza_Monniaux_LMCS12}, the correctness of the improvement
strategy (Lemma~\ref{lemma:improve}) and the correctness of the
strategy value computation (proved in \S\ref{sec:new_strategy_value}),
with the remark\rronly{ that our new algorithm can be considered an
  instance of the path-focused iteration scheme: the difference with
  the instance described in \S\ref{sec:original_strategy_iteration} is
  that we store $\rho$ in an efficient way and the way we pick the
  improved strategy, neither of which matters for correctness}.

Strategy iteration terminates in at most as many iterations as there
are possible strategies: here, for each of the $2^n$ states $\vec{b}'$
and $i$-th constraint ($1 \leq i \leq m$), there are $2^n$ possible
choices for $\sigma(i,\vec{b}')$ and $2^d$ choices for
$\pi(i,\vec{b}')$, thus the bound.
\end{proof}}
\pponly{\vspace*{-4ex}}

\subsection{Computing the Strategy Value with Fewer Unknowns}
\label{sec:new_strategy_value}

There remains the second difficulty: computing the value of a given
strategy, that is, computing $\rho(\vec{b})$ for $\vec{b} \in \BB^n$,
thus solving linear programs with at least $m2^n$
variables~\cite[\S6.4]{Gawlitza_Monniaux_LMCS12}.  We solve this
difficulty by remarking that $\rho(i,\vec{b})$ is the same for all
$\vec{b}$ in the same equivalence class with respect to $\sim_i$:
$\vec{b}_1 \sim_i \vec{b}_2 \iff \pi(i,\vec{b}_1) = \pi(i,\vec{b}_2)
\land \sigma(i,\vec{b}_1) = \sigma(i,\vec{b}_2)$.  Assuming $\vec{b}
\mapsto \sigma(i,\vec{b})$ and $\vec{b} \mapsto \pi(i,\vec{b})$ are
stored as \textsc{MtBdds}\rronly{ (or, equivalently, $n$ ordinary
  \textsc{Bdd}s for $\vec{b} \mapsto \sigma(i,\vec{b})$ and $d$ for
  $\vec{b} \mapsto \pi(i,\vec{b})$, each containing a bit of the
  image)}, the equivalence classes are obtained as \textsc{Bdd}s using
the reverse images of these functions.

We then apply the algorithm from \S\ref{sec:original_value}, but
instead of the whole set of $\rho(i,\vec{b})$ unknowns for $\vec{b}
\in \BB^n$ and $1 \sleq i \sleq m$, we only pick one unknown $c_{i,j}$
per equivalence class; these unknowns define $\rho$ in the form
expected by the strategy improvement step of
\S\ref{sec:new_strategy_improvement}.
Remark that, if the equivalence classes are computed as \textsc{Bdd}s,
it is trivial to turn them into logical formulas $\phi_{i,j}$ of
linear size w.r.t. that of the \textsc{Bdd}.  Notice that also the
$\infty$-abstraction technique
\cite[\S8,9]{Gawlitza:2011:SSR:1961204.1961207} also applies.
Let $\eqvclass{\vec{b}}{i}$ denote the equivalence class of $\vec{b}$
with respect to $\sim_i$; $\pi$ directly maps from equivalence classes
as $\pi(i,\eqvclass{\vec{b}}{i}) \defn \pi(i,\vec{b})$
(resp. for~$\sigma$).

\begin{example*}
Let us assume the current abstract value \\[0.6ex]
\centerline{$
\begin{array}{lll}
\rho(1)=\{\neg\verr \rightarrow \frac{365}{16}, &&\verr\rightarrow -\infty\}\\
\rho(2)=\{\neg\verr\wedge\vheat \rightarrow -16, &\neg\verr\wedge \neg\vheat \rightarrow -22, &\verr\rightarrow -\infty\}.
\end{array}
$}
Moreover, assume that we have computed the following strategy for the
first template row:
$\sigma(1,\bar{\verr}\bar{\vheat}\vfan)=\sigma(1,\bar{\verr}\bar{\vheat}\bar{\vfan})\in\{\bar{\verr}\bar{\vheat}\vfan,\bar{\verr}\bar{\vheat}\bar{\vfan}\}$ and
$\pi(1,\bar{\verr}\bar{\vheat}\vfan)=\pi(1,\bar{\verr}\bar{\vheat}\bar{\vfan})=(1,0,0)$. 
Then the states $\bar{\verr}\bar{\vheat}\vfan$ and $\bar{\verr}\bar{\vheat}\bar{\vfan}$ will
be in the same equivalence class, because both bounds will have the same
value in the strategy fixed point. Hence, we have to generate only one
set of constraints for both states when solving the LP problem that
characterizes the strategy fixed point $\rho$. 

\noindent We finally obtain\footnote{By maximizing $-t$ for any $te$
  in $t\geq 18 \wedge 14\leq te\leq 19 \wedge t' = \frac{15t+te}{16}$, 
we get $-\frac{15\cdot 18+14}{16}=-\frac{71}{4}$.} 
$
\left\{\begin{array}{lll}
\rho(1)=\{\neg\verr \rightarrow \frac{365}{16}, &&\verr\rightarrow -\infty\}\\
\rho(2)=\{\neg\verr\wedge\vheat \rightarrow -16, &\neg\verr\wedge \neg\vheat \rightarrow -\frac{71}{4}, &\verr\rightarrow -\infty\}.
\end{array}\right.
$
This is actually the strongest inductive abstract
invariant of our program: $\neg\verr\wedge\vheat \wedge 16\leq t\leq
\frac{365}{16} \vee \neg\verr\wedge\neg \vheat \wedge \frac{71}{4}\leq t\leq
\frac{365}{16}$.
\end{example*}

\begin{theorem}
  Let $\abstr{\rho}$ be the result of the modified strategy evaluation
  and $\rho_{k+1} = \mu_{\geq \rho_k} \Psi_{\pi_{k+1}}$ be the result
  of the original strategy evaluation. Then for all $i,\vec{b}$,
  $\rho(i,\vec{b}) = \abstr{\rho}(i,\eqvclass{\vec{b}}{i})$.
\end{theorem}
\rronly{
\begin{proof}
  The original strategy evaluation computes $\rho_{k+1} = \mu_{\geq \rho_k} \Psi_{\pi_{k+1}}$. Remark that $\rho_{k+1}$ is thus the limit of the ascending sequence $u_0 = \rho_k$, $u_{j+1} = \Psi_{\pi_{k+1}} (u_j)$; furthermore, from the definition of $\Psi_{\pi_{k+1}}$ and the form of the equivalence classes, for any $j \geq 1$, $u_j(i,\vec{b})$ does not depend on the choice of $\vec{b}$ in an equivalence class of~$\sim_i$. It follows that the same limit is obtained by keeping for each $j$ only one $u_j(i,\vec{b})$ per equivalence class. This corresponds to iterating\\[-2ex]
\begin{equation}
\hspace*{-1em}\abstr{\Psi}_{\pi}(\rho) \defn (i',\eqvclass{\vec{b}'}{i'}) \mapsto
  \sup \left\{ \mat{A}_i \vec{x}' \mid \exists \vec{x},\vec{y}.~ T_{\pi(i',\vec{b}')}
          \land (\mat{A}\vec{x} \leq \rho(i',\sigma(\eqvclass{\vec{b}'}{i'}))) \right\}
\end{equation}

As in \S\ref{sec:original_value}, the modified strategy evaluation
computes the least fixed point $\abstr{\rho}$ of
$\abstr{\Psi}_{\pi_{k+1}}$ greater than $(i,\eqvclass{\vec{b}}{i})
\mapsto u_1(i,\vec{b})$. But, from the remark above, this implies that
for all $(i,\vec{b})$, $\rho_{k+1}(i,\vec{b}) =
\abstr{\rho}(i,\eqvclass{\vec{b}}{i})$.
\end{proof}}

\subsection{Abstraction Through Limitation of Partitioning}
Even though we have taken precautions against unnecessarily large
numbers of unknowns by grouping ``equivalent'' Boolean states
together, it is still possible that the number of equivalence classes
to consider grows too much as the algorithm proceeds. 
It is however possible to freeze them
permanently\rronly{, for instance} to their last sufficiently small value.
Only small modifications to the algorithms are necessary: The strategy
value computation (\S\ref{sec:new_strategy_value}) remains the same
except that the equivalence classes are never recomputed. Let
$\phi_{i,1},\dots,\phi_{i,s_i}$ denote the propositional formulas (in
$\vec{b}$) defining the equivalence classes with respect to constraint
number~$i$.  In the strategy improvement step
(\S\ref{sec:new_strategy_improvement}) $\sigma(i,j) \in \BB^n$
(resp. $\pi(i,j)$) is now defined for the index $1 \sleq j \sleq s_i$ of
an equivalence class with respect to constraint~$i$.

\rronly{ 
The correctness proofs stay the same, except that instead of
  computing the least fixed point in $\left(\{ 1, \dots, \ell \}
    \times \BB^n \right)\rightarrow \overline{\QQ}$ we compute it in
  $\left( \bigsqcup_{1, \dots, \ell} E_i \right) \rightarrow
  \overline{\QQ}$ where $E_i$ is the set of equivalence classes
  associated with constraint~$i$; the latter lattice is included in
  the former.  
}

\rronly{
\subsection{Combination with Predicate Abstraction}
We have described so far a method for computing template polyhedral invariants on each element of a partition of the state space according to the value of Booleans $b_1,\dots,b_n$. These Booleans may be replaced by arbitrary predicates $\chi_1,\dots,\chi_n$: it suffices to replace $T$ by $T \land \bigwedge_i (b_i \Leftrightarrow \chi_i)$.

\subsection{Strategy Iteration with Partitioning is EXPTIME-hard}

In preceding work without partitioning
\cite{GM11,Gawlitza_Monniaux_LMCS12}, the single-exponential upper
bound was shown to be reached by a contrived example program, and the
decision problem associated with the least invariant computation
(``given a template, a transition relation, an initial state and a bad
state, is there an inductive invariant that excludes the bad state'')
was shown to be $\Sigma^p_2$-complete.  We have an \textsc{nexptime}
upper bound on the problem.  We will now prove
\textsc{exptime}-hardness for the problem with partitioning.

Let $\Pi$ be an \textsc{exptime} problem.  Consider a Turing machine $\calM$
deciding $\Pi$, with a single tape over the alphabet $\{0,1\}$ with
time bounded by $2^{P(n)}$ and finite state in~$\Sigma$.  Let $x$ be
an input of size $n$ to $\calM$; we are going to describe a program
$\calP$ of length proportional to $P(n)$ such that its execution would
yield the same result as running $\calM$ over~$x$. $\calP$ only uses
Boolean operations for discrete state, and affine linear operations
for continuous state.

Let $\calM_x$ be the Turing machine $\calM$ where $x$ has been
substituted into the input; $|\calM_x| \simeq |\calM| + |x|$.  The
tape is modeled as a couple of natural integers $(l,r)$ with
$2^{P(n)}$ bits, where
\begin{itemize}
\item $l$ represents the bits strictly to the left of the read-write
  head, the $2^{P(n)}-1$-th bit representing the bit on the tape just
  left of the read-write head, and the least order bit representing
  the bit on the tape $2^{P(n)}$ positions left of the head;
\item $r$ represents the bits to the right of the read-write head, the
  $2^{P(n)}-1$-th bit representing the bit on the tape under the
  read-write head, and the least order bit representing the bit on the
  tape $2^{P(n)}-1$ positions right of the head.
\end{itemize}
$l$ and $r$ are initialized to $0$ (empty tape).

A step of the Turing machine $\calM_x$ is simulated as follows:
\begin{itemize}
\item the bit $b$ under the read-write head is obtained by taking the
  $2^{P(n)}-1$-th bit of $r$, by comparing $r$ to $K =
  2^{2^{P(n)}-1}$;
\item the bit $w$ just to the left of the read-write head is similarly
  obtained by comparing $l$ to~$K$;
\item the bit $b'$ to be written to the tape under the head, the
  direction of movement of the head and the next state are computed
  according to the rules of~$\calM$;
\item if the tape is to be moved to the left, then a parallel update
  is made: $l:=2(l-wK)$ and $r:=r/2+b'K$
\item if the tape is to be moved to the right, then a parallel update
  is made: $r:=2(l-bK)$ and $l:=l/2+b'K$
\end{itemize}

Since there are $2^{P(n)}$ steps to be simulated, we loop over the
simulated step using a binary counter $S$ with $P(n)$~bits. This loop
ends when the Turing machine under simulation enters a final state.

Note that, in the program, $l$ and $r$ are always natural
numbers. This is because, when we execute $l/2$ (resp. $r/2$), the
low-order bit of $l$ (resp.~$r$) is necessarily $0$, otherwise it
would mean that $\calM$ would be using more than $2^{P(n)}$ bits of
tape. Also, the operations $wK$, $bK$, $b'K$ are defined not by
non-linear multiplications, but by case analysis over the bits $w$,
$b$ and $b'$. All resulting elementary operations are thus linear over
the reals.

The last issue to solve is how to create $K$. We cannot write it as a
constant in the program, because it has $2^{P(n)}$ bits --- the
program would have exponential size. Instead, we prepend to the
program $K:=1$ followed by a sequence of $2^{P(n)}-1$ doublings ($K :=
2K$), implemented by a loop over binary counter $S$ of length $P(n)$.

We thus obtain a program of length $O(P(n)+n)$ (because of the
operations over binary counter $S$ of length $P(n)$). It has $P(n)+
2\lceil \log_2 |\Sigma|\rceil + 3$ bits of discrete storage (not
counting control flow: $3$ for $b,b',w$, $P(n)$ for the binary counter
$S$, $2 \log_2 |\Sigma|\rceil$ for implementing the state transition).

Now note the execution of $\calP$ is fully deterministic. In addition,
along an execution trace $(p,S)$, where $p$ is the control point in
$\calP$ and $S$ the binary counter, takes distinct values ($S$ is
incremented once per loop iteration and $p$ follows control inside the
loop).  Thus, an interval analysis that has a different interval for
$l$ and $r$ for each value of $(p,S)$ will essentially simulate the
concrete execution of~$\calP$ and obtain an \emph{exact} result.  Such
an interval analysis can thus decide whether $\calP$ terminates in
``accepting'' or ``rejecting'' answers, and thus whether $\calM$
accepts or rejects~$x$.


Despite repeated attempts, we have not yet been able to narrow the
interval at proving \textsc{nexptime}-completeness. It is thus
possible that worst-case complexity is actually better.  
}

\section{Experiments}

We have prototypically implemented the algorithm in the static
analyzer \textsc{ReaVer} \cite{Sch12} (\rronly{written in OCaml and
}taking \textsc{Lustre} code as input) using the LP solver
\textsc{QSOpt\_Ex}\footnote{version 2.5.6,
  \url{http://www.dii.uchile.cl/~daespino/ESolver_doc/main.html}}, the
SMT solver \textsc{Yices}\footnote{version 1.0.40,
  \url{http://yices.csl.sri.com/}} and the BDD package
\textsc{Cudd}\footnote{version 2.4.2,
  \url{http://vlsi.colorado.edu/~fabio/CUDD/}}.  The implementation
makes heavy use of incremental SMT solving.

\paragraph{Tested variants of the algorithm.}
We implemented the following variants of the algorithm to compare
their performance:
\begin{compactitem} 
\item[(n)] Naive model enumeration using SMT solving per template row
  as explained in the first part of \S\ref{sec:new_strategy_improvement}. This
  corresponds to updating $\pi$ and $\sigma$ in
  Alg.~\ref{algo:improve} using the model obtained in line 5
  ($G=(\vec{b}'=\instance{\vec{b}}')$) without doing lines 6 to 11 and 15.
\item[(t)] Enhancement of (n) by trying to reapply successfully improving models to other template rows.
\item[(s)] Symbolic encoding of template rows and model enumeration
  over the whole template at once, \ie, lines line 2 and 17 are omitted
  because the template row $i'$ becomes part of the SMT formula to be
  solved for in line 4, and is then retrieved from the model returned in line 5.
\item[(g)] Alg.~\ref{algo:improve} with \emph{generalization} as
  described in \S\ref{sec:new_strategy_improvement}, but without the
  inner iterations (\ie, without lines 7 to 9 and 15) that search for
  models of the purely propositional formula $F$.
  Hence, (g) obtains the models to be generalized from the SMT formula
  in line 4 only.
\item[(m)] Alg.~\ref{algo:improve} as given.
\end{compactitem}
All these variants reduce the number of unknowns in the LP problem
using equivalence classes (see \S\ref{sec:new_strategy_value}).
Furthermore, we used an implementation of the original max-strategy
algorithm \cite{GS07} (GS07), and the improvements using
SMT solving proposed in \cite{GM11} (GM11).  Note that these latter
two algorithms need to enumerate $\mathcal{O}(2^n)$ control states
(where $n$ is the number of Booleans in the recurrent state).
Yet, the size of the control flow graphs (see Table~\ref{tab:exp2})
generated using the method described in \cite[\S7.3]{Sch12}
is often far smaller than the worst case.
The difference between GS07 and GM11 is essentially that, for each
template row, the former tests all strategies to find an improvement,
whereas the latter asks the SMT solver to find an improving strategy
in the disjunction of available strategies.

It is important to note that all these variants of the algorithm
return the \emph{same} invariants, \ie~the strongest invariants in the
domain $\BB^n \rightarrow A$ where $A$ is a given template abstract
domain. The only difference is the way the strategy improvement is
computed.

\paragraph{Comparisons.}
We performed two kinds of comparisons:\footnote{The examples and
  detailed experimental results can be found on
  \url{http://www.cs.ox.ac.uk/people/peter.schrammel/reaver/maxstrat/}.} 

\begin{compactenum} 
\item We evaluated the scalability of various variants of the
  max-strategy improvement algorithm on small benchmarks (1-, 2-, and
  3-dimensional array traversals, parametric in size by duplicating
  functionality and adding Boolean variables) that exhibit the
  strategy and state space explosion expected to occur in larger
  benchmarks. We used box and octagonal templates, giving a total of
  96 benchmarks. 
\item We compared the max-strategy improvement algorithm with standard
  forward analysis with widening and with abstract acceleration
  \cite{GH06,SJ12c} (both using widening after two iterations and
  applying two descending iterations\footnote{We did not not observe
    any improvement in precision beyond these values.}) on reactive
  system models (traffic lights \cite{BKA08}, our thermostat, car
  window controller \cite{SMK13a}, and drug pump \cite{SHL11}), again
  deriving the more complex variants 2 and 3 by adding and duplicating
  functionality (e.g. branching \emph{multiple} drug pumps to a
  patient and checking the concentration in the blood). 
\end{compactenum}

\paragraph{Results.}
The first comparison (see Fig.~\ref{fig:exp1}) shows that
the various variants of the algorithm behave quite differently in
terms of runtime:
The GM11 improvement is on average 22\% faster than the original
algorithm.
(t) and (s) scale better than (n). It is interesting to observe that
(t) and (s) perform similarly although their algorithms are very
different.
The most important optimization of the strategy improvement algorithm
proposed in this paper is the generalization step which makes it scale
several orders of magnitude better than the other variants,
because it avoids naive model enumeration.
The results indicate that the full Alg.~\ref{algo:improve}
(variant (m)) is slower than the variant (g) without the innermost
iterations. A possible explanation for this is that as soon as all
models have been enumerated, (m) has to confirm 
unsatisfiability by checking both $F\wedge\neg G$ and $\psi\wedge\neg U$.

\rronly{
However, a broader evaluation is necessary to come to definite
conclusions since the structure of the benchmarks is quite simple.
}

\begin{figure}[t]
\centering
\vspace*{-4ex}
\includegraphics{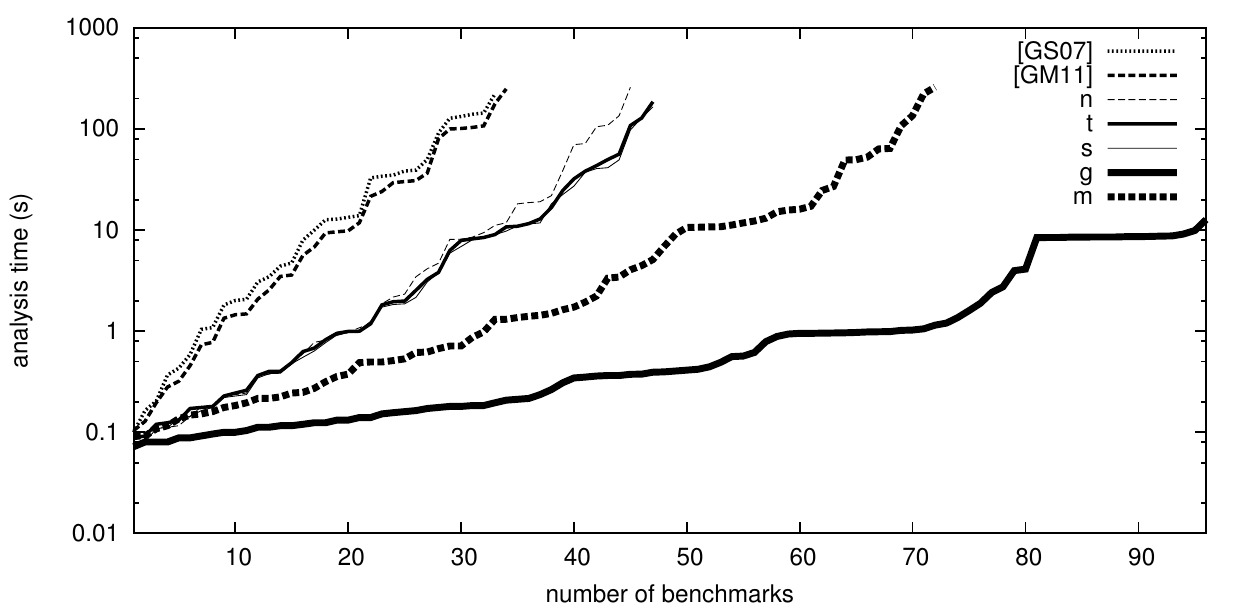}
\vspace*{-5ex}
\caption{\label{fig:exp1}
Comparison of various variants of the max-strategy improvement
algorithm. All these algorithms compute the \emph{same} invariant. 
}
\end{figure}

The results of the second comparison (see Table~\ref{tab:exp2})
indicate that max-strategy iteration is able to compute
better invariants than techniques relying on widening in the \emph{same}
$\BB^n \rightarrow A$ abstract domain.
Enhanced widening techniques, such as abstract acceleration, do
occasionally improve on precision, but without guarantee to find
the best invariant. 

We emphasize again that all four max-strategy iteration algorithms in
Table~\ref{tab:exp2} compute identical invariants.

An open problem w.r.t. all template-based analysis techniques is
however the generation of good templates.  For our experiments, we
have chosen the weakest of the standard templates (boxes, zones,
octagons) that can express the required invariant.
Strategy iteration is in general the more expensive technique, but due
to our improvements the performance is pushing forward into a
reasonable range.  These results also show that variant (g) --
although a bit slower than (s) in many cases -- seems to scale best.

\begin{table}[t]
  \scriptsize
\centering
\vspace*{-3ex}
\hspace*{-3em}
\begin{tabular}{|l|c|rrrr|rr|r@{}lc|r@{}lc|r@{}lc|r@{}lc|r@{}lc|r@{}lc|r@{}lc|}
\hline
& &\multicolumn{6}{|c|}{size} & \multicolumn{6}{|c|}{previous algorithms} & \multicolumn{6}{|c|}{this paper} & \multicolumn{3}{|c|}{std.} & \multicolumn{3}{|c|}{abstr.} \\
& &\multicolumn{4}{|c|}{vars}&\multicolumn{2}{|c|}{CFG} & \multicolumn{3}{|c|}{GS07\cite{GS07}} & \multicolumn{3}{|c|}{GM11\cite{GM11}} & \multicolumn{3}{|c|}{g} & \multicolumn{3}{|c|}{s} & \multicolumn{3}{|c|}{analysis} & \multicolumn{3}{|c|}{accel.} \\
&dom&b&n&bi&ni&lc&ed & \multicolumn{2}{|c|}{time} & p & \multicolumn{2}{|c|}{time} & p &\multicolumn{2}{|c|}{time} & p &\multicolumn{2}{|c|}{time} & p &\multicolumn{2}{|c|}{time} & p &\multicolumn{2}{|c|}{time} & p \\
\hline
Traffic 1 &B&6&6&0&0&18&61& 2&.16 &$\checkmark$& 2&.10 &$\checkmark$& 
2&.33 &$\checkmark$& 2&.16 &$\checkmark$& 1&.22 &$\checkmark$&\bf 0&\bf.43 &$\checkmark$\\
Traffic 2 &Z&6&8&0&0&18&151& 122& &$\checkmark$& 114& &$\checkmark$& 
108& &$\checkmark$&\bf 97&\bf.0 &$\checkmark$& 3&.49 && 2&.86* &\\
Traffic 3 &Z&8&8&1&0&50&619& 674& &$\checkmark$& 640& &$\checkmark$& 
357& &$\checkmark$&\bf 329& &$\checkmark$& 22&.1 && 19&.2* &\\
Thermostat 1&B&4&3&0&2&6&15& 0&.36 &$\checkmark$& 0&.32 &$\checkmark$& 
0&.28 &$\checkmark$&\bf 0&\bf.26 &$\checkmark$& 0&.82 && 0&.85 &\\
Thermostat 2&B&6&5&0&4&18&145& 16&.8 &$\checkmark$& 15&.1 &$\checkmark$& 
3&.44 &$\checkmark$&\bf 3&\bf.23 &$\checkmark$& 26&.6 && 30&.4 &\\
Thermostat 3&B&8&7&0&6&66&1357& 720& &$\checkmark$& 715& &$\checkmark$& 
66&.5 &$\checkmark$&\bf 61&\bf.9 &$\checkmark$& 674& && 908& &\\
Window 1&O&9&5&5&0&21&120&109& &$\checkmark$& 102& &$\checkmark$& 
\bf 70&\bf.7 &$\checkmark$& 73&.4 &$\checkmark$& 4&.57 && 4&.70 &\\
Window 2&O&11&5&6&0&45&452& 394& &$\checkmark$& 372& &$\checkmark$& 
\bf 189& &$\checkmark$& 286& &$\checkmark$& 18&.57 && 23&.5 &\\
Window 3&O&13&5&7&0&81&1388& 1412& &$\checkmark$& 1220& &$\checkmark$& 
\bf 242& &$\checkmark$& 697& &$\checkmark$& 70&.2 && 93&.5 &\\
DrugPump 1 &B&4&10&4&1&6&231& 92&.6 &$\checkmark$& 90&.3 &$\checkmark$& 
6& .05 &$\checkmark$&\bf 4&\bf.55 &$\checkmark$& 210& && 120& &\\
DrugPump 2 &B&7&12&8&1&34&11201& \multicolumn{6}{|c|}{timeout $>1800$} & 
149& &$\checkmark$& \bf 95&\bf.5 &$\checkmark$&
\multicolumn{6}{|c|}{timeout $>1800$}\\
DrugPump 3 &B&10&14&8&1&146&112561& \multicolumn{6}{|c|}{timeout $>1800$} & 
\bf 1019& &$\checkmark$& 1396& & $\checkmark$& \multicolumn{6}{|c|}{timeout $>1800$}\\
\hline
\end{tabular}~\\[2ex]
 \caption{\label{tab:exp2}
Comparison of max-strategy iteration with standard analysis approaches
(dom: domain used (boxes (B), zones (Z), octagons (O)); number of
variables: Boolean (b), numerical (n), Boolean and numerical inputs
(bi, ni); number of locations (lc) and edges (ed) of the control flow
graph (CFG); analysis time in seconds; property proved (p); fastest in
bold). (* computed with octagons, because zones are not available)
}
\end{table}

\pponly{\vspace*{-2ex}}
\section{Related Work}
\pponly{\vspace*{-1ex}}
It has long been recognized that it is a good idea to distinguish
states according to Boolean variables or arbitrary predicates
(as in \emph{predicate abstraction}). Yet, taking all Boolean
variables into account tends to be unbearably expensive\rronly{: checking the
reachability of a state for a purely Boolean program is
\textsc{pspace}-complete, in practice often solved by constructing a
\textsc{Bdd} describing reachable states as the result of a least
fixed point computation, which may have exponential size}.
\rronly{

While it is possible to encode finite-precision arithmetic into a
Boolean program, the large number of Boolean variables and the
complicated transition structure generally result in poor performance,
thus the incentive to separate arithmetic from ``true'' Booleans and
other small enumerated types.  Note that this may not be so obvious
for languages such as C or intermediate representations such as LLVM
bitcode, where such types may be encoded as integers; a pre-analysis
may be necessary~\cite{DBLP:journals/entcs/JeannetS12}.

Even if the number $n$ of Boolean variables has been suitably
reduced, distinguishing all combinations may be too costly.} Various
heuristics have therefore been proposed so as to partition $\BB^n$
into a reasonably small number of subsets
\cite{DBLP:conf/sas/SchrammelJ11}. Relations between the Boolean and
numerical states are only kept w.r.t. these equivalence
classes~\cite{DBLP:conf/vmcai/SchrammelS13}.
Combining the latter technique with the method presented in this paper
to limit partitioning would certainly improve efficiency, however, to
the detriment of precision of the obtained invariant which strongly
depends on the choice of a clever partitioning heuristics.

Early work in compilation and verification of reactive systems
\cite{DBLP:journals/scp/BouajjaniFHR92} advocated quotienting the
Boolean state space according to some form of \emph{concrete}
bisimulation. In contrast, we compute coarser equivalences according
to per-constraint \emph{abstract} semantics.
In the \rronly{industrial-strength analyzer }\textsc{Astr\'ee}\pponly{
  analyzer}, static heuristics determine reasonably small packs of
``related'' Booleans and numerical variables, such that the values of
the numerical variables are analyzed separately for each Boolean
valuation~\cite[\S6.2.4]{BlanchetCousotEtAl_PLDI03}. In contrast, our
equivalence classes are computed dynamically and per-constraint.

\rronly{
\emph{Disjunctive invariants} are related to the partitioning
approach; in both cases the invariant is a disjunction $C_1 \lor \dots
\lor C_d$ where $C_i$ are simpler invariants (typically, conjunctions
of certain types of literals), but in the disjunctive invariant
approach the $C_i$ may overlap (that is, not have pairwise empty
intersection). In such a system, union (as at control merge points)
may be implemented by simple concatenation of the disjunctions, but
this quickly leads to a blowup; instead a criterion could be used to
merge those $C_i$, \eg, of which the abstract union is actually exact;
a similar problem occurs with widening
operators~\cite{DBLP:journals/sttt/BagnaraHZ06}. An alternative, which
bears some limited resemblance to our strategy-based approach, is to
build a map $\sigma$ meaning that disjunct $C_i$ flows through path
$\pi$ into disjunct
$C_{\sigma(i,\pi)}$~\cite{Henry_Monniaux_Moy_SAS2012}.
}

The strategy iteration we have applied proceeds ``upward'', by
successive under-approximations of the least inductive invariant%
\rronly{ inside the domain converging to it in a finite number of
  iterations}; strategies correspond to \rronly{paths inside the
  program, which map }to ``max'' operators in a high-level vision of
the problem. There also exists ``downward'' strategy iteration, where
strategies correspond to ``min'' operators\rronly{ (tests inside the
  programs and internal reductions of the abstract domain)}:
iterations produce successive \emph{over}-approximations of the least
inductive
invariant~\cite{DBLP:conf/esop/GaubertGTZ07,DBLP:conf/atva/SotinJVG11},
to which convergence is ensured in some cases. A bonus of such an
approach is that each iteration produces an over-approximation of the
least inductive invariant\rronly{ inside the domain}, which may be
used to prove safety properties without having to wait for
convergence. Sadly, it does not seem to be easily adapted to
approaches based on SMT solving, since the SMT formulas would contain
universal quantifiers\rronly{, which greatly complicates their solving}.

\pponly{For a more comprehensive discussion of related work, we refer to 
the extended version \cite{MS14a}.}
\rronly{
Recently, a tool for \emph{optimization modulo theory}
was presented \cite{LAK+14}. We plan to test the
variant of our algorithm maximizing $\Delta$ (see
\S\ref{sec:new_strategy_improvement}) with the help of this tool.
}

\pponly{\vspace*{-1.5ex}}
\section{Conclusion}
\pponly{\vspace*{-0.5ex}}
We propose a method for computing strongest invariants in
linear template domains when the control states are
partitioned according to $n$ Booleans or arbitrary predicates, thereby
producing a combination of predicate abstraction and template
polyhedral abstraction.  \rronly{In accordance with preceding works
\cite{Henry_Monniaux_Moy_SAS2012,Gawlitza_Monniaux_LMCS12,Monniaux_Gonnord_SAS11},
it traverses loop-free parts of the control graph without need for
intermediate abstraction, thus improving the precision.} Our method
performs strategy iteration, and dynamically partitions 
the states according to an equivalence relation depending on the
current abstraction at each step. The final result is optimal in
the sense that it is the strongest invariant in the abstract domain{\rronly,
which a naive algorithm would obtain in at least exponential time and
space}. \rronly{While an upper bound on the number of equivalence classes in
our algorithm is also exponential $n$, it can be limited arbitrarily,
with some loss of precision. The upper bound on the number of
iterations is doubly exponential in~$n$.}
Our experimental results demonstrate the significant
performance impact of various optimizations and
the ability to compute more precise invariants in comparison to
widening-based techniques.

\pponly{
In preceding work without partitioning
\cite{GM11,Gawlitza_Monniaux_LMCS12}, the single-exponential upper
bound was shown to be reached by a contrived example
program, and the decision problem associated with the least invariant
computation (``given a template, a transition relation, an initial
state and a bad state, is there an inductive invariant that excludes
the bad state'') was shown to be $\Sigma^p_2$-complete. In contrast,
although we have an \textsc{nexptime}
upper bound and proved \textsc{exptime}-hardness
(see \cite{MS14a} for a proof)
for the problem with partitioning,
we have not yet been able to \rronly{narrow the
interval despite repeated attempts at proving}\pponly{prove}
\textsc{nexptime}-completeness\rronly{. It is thus possible
  that }\pponly{---thus} the worst-case
complexity \pponly{could}\rronly{is} possibly \pponly{be }better.
}

\bibliographystyle{splncs03}
\bibliography{powerset_policies}
\end{document}